\documentclass[10pt,a4paper]{article}
\usepackage{amsfonts}
\usepackage{graphics,amssymb,amsthm,amsmath,epsf}
\usepackage{graphicx}
\usepackage[round]{natbib}
\usepackage{color}
\newtheorem{lemma}{Lemma}[section]

\numberwithin{equation}{section}

\allowdisplaybreaks
\begin{document}

\date{}
\title{Comparing Two Approaches in Heteroscedastic \\Regression Models}
\author{A. A. Jafari$\thanks{aajafari@yazd.ac.ir}$ \\
{\small Department of Statistics, Yazd University, Yazd, Iran}}

\date{}
\maketitle

\begin{abstract}
Recently,  a generalized test approach is proposed by
\cite{Sadooghi-ja-ma-16}
and a fiducial approach is proposed by
\cite{xu-li-18}
to test the equality of coefficients in several regression models with unequal variances. In this paper, it is shown that the considered quantities in these approaches are identically distributed and therefore, these approaches are same. Also, this result satisfies for the one-way ANOVA problem.
\end{abstract}


\section{Introduction}

Consider the $k$ regression models
\[{{\boldsymbol Y}}_i =X_i{{\boldsymbol \beta }}_i+{{\boldsymbol \varepsilon }}_i\ \ \ \ {{\boldsymbol \varepsilon }}_i~\sim N\left({\boldsymbol 0}{\boldsymbol ,\ }{\sigma }^2_iI_{n_i}\right),\ \ \ \ \ \ \ i=1,\dots ,k,\]
where $X_i$ is $n_i\times p$ design matrix with rank $p$, ${{\boldsymbol Y}}_i{\boldsymbol =(}Y_{i1},\dots ,Y_{{in}_i})'$ is $n_i\times 1$ ($n_i>p$ for all $i$) observation vector, ${{\boldsymbol \beta }}_i=({\beta }_{i1},{\beta }_{i2},\dots ,{\beta }_{ip})'$ is the vector of parameters with dimension of $p$, ${\boldsymbol \varepsilon}_i$ is the $n_i\times 1$ disturbance vector, and $I_{n_i}$ is $n_i\times n_i$ identity matrix. Furthermore, all the ${{\boldsymbol \varepsilon }}_i$ are independent.

It is well-known that the unbiased estimations for ${{\boldsymbol \beta }}_i$ and ${\sigma }^2_i$ are ${\hat{{\boldsymbol \beta }}}_i ={\left(X'_iX_i\right)}^{-1}X'_i{{\boldsymbol Y}}_i$ and
$S^2_i=\boldsymbol Y'_i\left(I_p-X_i{\left(X'_iX_i\right)}^{-1}X'_i\right){\boldsymbol Y}_i/(n_i-p)$, respectively, such that they are independent. When the variances ${\sigma }^2_i$'s are unknown, an usual test statistic to test
\begin{equation}\label{eq.H0}
H_0:{{\boldsymbol \beta }}_1={{\boldsymbol \beta }}_2=\dots ={{\boldsymbol \beta }}_k,
\end{equation}
is
\citep[see][]{tian-ma-ve-09,Sadooghi-ja-ma-16,xu-li-18}
\[Q_0=\sum^k_{i=1}S^{-2}_i\hat{\boldsymbol \beta}'_i\left(X'_iX_i\right) \hat{\boldsymbol \beta}_i-
\left[\sum^k_{{\boldsymbol i}=1}{S^{-2}_i\hat{\boldsymbol \beta }'_i\left(X'_iX_i\right)}\right]{\left[\sum^k_{i=1}{S^{-2}_i\left(X'_iX_i\right)}\right]}^{-1}\left[\sum^k_
{i=1}S^{-2}_i\left(X'_iX_i\right)\hat{\boldsymbol \beta }_i\right].\]

Under null distribution, $Q_0$ can be approximated by the chi-square distribution with $p (k-1)$ degrees of freedom for large sample sizes. However, this approximation does not work well for small samples
\citep[see][]{tian-ma-ve-09,Sadooghi-ja-ma-16,xu-li-18}.
Therefore, some other approaches are proposed to test equality of regression models with unequal variances for example a parametric bootstrap approach by
\cite{tian-ma-ve-09},
a generalized approach  by
\citep{Sadooghi-ja-ma-16}
and a fiducial approach by
\citep{xu-li-18}.
In this paper, we compare these generalized and fiducial approaches. We will see that although these approaches are proposed in different ways and are not the same in appearance but they are identical.

 In Section \ref{sec.com}, the generalized and fiducial approaches are reviewed and compared. In Section \ref{sec.anova}, one way-ANOVA is discussed as a special case.

\section{Comparing two approaches}
\label{sec.com}

For testing $H_0$ in \eqref{eq.H0}, a generalized approach is proposed by
\cite{Sadooghi-ja-ma-16}
and a fiducial approach is proposed by
\cite{xu-li-18}.
 In this section, these  approaches are reviewed, briefly. Then, it is shown that they are identical.

Consider ${\boldsymbol b}_i$ and $s^2_i$, $i=1,\dots,k$, are the observed values of $\hat{\boldsymbol \beta}_i$ and $S^2_i$, respectively. For given $\left(\boldsymbol b_i, s^2_i\right)$,
\cite{xu-li-18}
derived a fiducial quantity to test \eqref{eq.H0} as
\begin{equation}\label{eq.QF}
Q_F=\sum^k_{i=1}{\boldsymbol t}'_i{\boldsymbol t}_i-\left[\sum^k_{i=1}{s^{-1}_i{{\boldsymbol t}}'_i{\left(X'_iX_i\right)}^{\frac{1}{2}}}\right]{\left[\sum^k_{i=1}{s^{-2}_i\left(X'_iX_i\right)}\right]}^{-1}
\left[\sum^k_{i=1}{s^{-1}_i{\left(X'_iX_i\right)}^{\frac{1}{2}}{\boldsymbol t}_i}\right],
\end{equation}
where ${\boldsymbol t}_i$ follows a multivariate student's t-distribution $t_p(n_i-p,{\boldsymbol 0},I_p)$, $i=1,\dots ,k$,
\citep[see][]{kotz-na-04multi}
and they are mutually independent. Therefore, the p-value to test \eqref{eq.H0} is given by $P(Q_F>Q_0)$.

Let $H=C \bigotimes I_p$ where $\bigotimes $ denotes Kronecker product, $C=\left[I_{k-1}:{\boldsymbol 1}\right]$, and ${\boldsymbol 1}=(1,\dots ,1)'$.  Also, consider $W'W={\left[HSH'\right]}^{-1}$ where $S=\left[{\rm diag}\left(s^2_i{\left(X'_jX_j\right)}^{-1}\right)\right]$.
\cite{Sadooghi-ja-ma-16}
defined a generalized test variable to test the hypothesis in \eqref{eq.H0} as
\begin{equation}\label{eq.QG}
Q_G={{\boldsymbol Z}}'WH\left[{\rm diag}\left(\frac{\left(n_i-p\right)s^2_i}{U_i}{\left(X'_jX_j\right)}^{-1}\right)\right]H'W'{\boldsymbol Z},
\end{equation}
where
${\boldsymbol Z} \sim N ({\boldsymbol 0},I_{p(k-1)} )$ and $U_i\sim{\chi }^2_{(n_i-p)}$, $i=1,\dots,k,$ such that ${\boldsymbol Z}$ and $U_i$'s are mutually independent. Therefore, the generalized p-value  is given by $P(Q_G>Q_0)$.

\begin{lemma}
$Q_F$ and $Q_G$ are identically distributed.
\end{lemma}
\begin{proof}
Consider $\boldsymbol V_i \sim N ({\boldsymbol 0},I_p )$ and $U_i\sim {\chi }^2_{(n_i-p)}$, $i=1,\dots ,k,$ such that $\boldsymbol V_i$ and $U_i$'s are mutually independent. Also, consider  $D =\left[{\rm diag}\left(\frac{n_i-p}{U_i}\right)\right]\bigotimes I_p$.  Based on
\cite{Sadooghi-ja-ma-16},
 $Q_G$ has the same distribution as
\[Q^*_G= {\boldsymbol V}'\left(I_{pk}-qq'\right)D\left(I_{pk}-qq'\right){\boldsymbol V},\]
where
\[{\boldsymbol V}=\left[ \begin{array}{c}
{{\boldsymbol V}}_1 \\
{{\boldsymbol V}}_2 \\
\vdots  \\
{{\boldsymbol V}}_k \end{array}
\right],\ \ \ q=\left[ \begin{array}{c}
q_1 \\
q_2 \\
\vdots  \\
q_k \end{array}
\right],\ \ \ \ q_i=\left[s^{-1}_i{\left(X'_iX_i\right)}^{1/2}\right]{\left[\sum^k_{j=1}{s^{-2}_j\left(X'_jX_j\right)}\right]}^{-1/2}.\]
Consider $P=\left(I_{pk}-qq'\right)$ and $D^*=\left[{\rm diag}\left(\sqrt{\frac{n_i-p}{U_i}}\right)\right]\bigotimes I_p$. Then
\begin{eqnarray*}
Q^*_G&=&{{\boldsymbol V}}'PD^*D^*P{\boldsymbol V} = {{\boldsymbol V}}'D^*PPD^*{\boldsymbol V}
\\
&=&{{\boldsymbol V}}'D^*PD^*{\boldsymbol V}\\
&=&{{\boldsymbol V}}'D^*\left(I_{pk}-qq'\right)D^*{\boldsymbol V}\\
&=&{{\boldsymbol V}}'D{\boldsymbol V}-{{\boldsymbol V}}'D^*qq'D^*{\boldsymbol V}\\
&=&\sum^k_{i=1}{\frac{n_i-p}{U_i}{{\boldsymbol V}}'_i{{\boldsymbol V}}_i}-\left[\sum^k_{i=1}{\sqrt{\frac{n_i-p}{U_i}}s^{-1}_i
{{\boldsymbol V}}'_i {\left(X'_iX_i\right)}^{\frac{1}{2}}}\right]{\left[\sum^k_{i=1}{s^{-2}_i\left(X'_iX_i\right)}\right]}^{-1}\\
&&\times
\left[\sum^k_{i=1}{\sqrt{\frac{n_i-p}{U_i}}s^{-1}_i{\left(X'_iX_i\right)}^{\frac{1}{2}}{{\boldsymbol V}}_i}\right]\\
&=&\sum^k_{i=1}{{{\boldsymbol T}}'_i{{\boldsymbol T}}_i}-\left[\sum^k_{i=1}{s^{-1}_i{{\boldsymbol T}}'_i{\left(X'_iX_i\right)}^{\frac{1}{2}}}\right]{\left[\sum^k_{i=1}{s^{-2}_i\left(X'_iX_i\right)}\right]}^{-1}
\left[\sum^k_{i=1}{s^{-1}_i{\left(X'_iX_i\right)}^{\frac{1}{2}}{{\boldsymbol T}}_i}\right],
\end{eqnarray*}
where ${{\boldsymbol T}}_i=\sqrt{\frac{n_i-p}{U_i}}{{\boldsymbol V}}_i\sim t_p(n_i-p,{\boldsymbol 0},I_p)$. Therefore, $Q_F$ and $Q_G$ are identically distributed.
\end{proof}

\section{ One Way ANOVA problem}
\label{sec.anova}
Let $Y_{i1},Y_{i2},\dots ,Y_{in_i}$ is a random sample from a normal distribution with mean ${\mu }_i$ and variance ${\sigma }^2_i$, $i=1,\dots ,k$. The problem of testing equality of means of these $k$ distributions, i.e.
\[H^*_0:{\mu }_1=\dots ={\mu }_k,\]
is well-known to one-way ANOVA. It is a special case of $H_0$ in \eqref{eq.H0} with $p=1$ and $X_i={\boldsymbol 1}$. Therefore, the fiducial quantity in \eqref{eq.QF} becomes to
\begin{equation}\label{eq.QF2}
Q_F=\sum^k_{i=1}{t^2_i}-\frac{{\left(\sum^k_{i=1}{\sqrt{\frac{n_i}{s_i}}t_i}\right)}^2}{\sum^k_{i=1}{\frac{n_i}{s^2_i}}},
\end{equation}
where  $t_i$ has a t distribution with $n_i-1$ degrees of freedom, which is the fiducial quantity introduced by
\cite{li-wa-li-11}.
When $p=1$ and $X_i={\boldsymbol 1}$, the generalized test variable in \eqref{eq.QG} becomes to
\begin{equation}\label{eq.QG2}
Q_G={{\boldsymbol Z}}'WC\left[{\rm diag}\left(\frac{\left(n_i-p\right)s^2_i}{n_iU_i}\right)\right]C'W'{\boldsymbol Z},
\end{equation}
where $W'W={\left[CSC'\right]}^{-1}$, $S=\left[{\rm diag}\left(s^2_i/n_i\right)\right]$,
${\boldsymbol Z}\sim N ({\boldsymbol 0},I_{(k-1)} )$ and $U_i\sim {\chi }^2_{(n_i-p)}$, $i=1,\ \dots ,k$. This generalized test variable is introduced by
\cite{sa-ja-ma-12}.
This generalized test variable is also proposed by
\cite{xu-wa-08-anova}
in another form
\citep[see][]{sa-ja-ma-12}.

Based on Lemma, $Q_F$ in
 \eqref{eq.QF2}
 and $Q_G$ in  \eqref{eq.QG2}
 are identically distributed. Therefore, the all proposed approaches by
  \cite{xu-wa-08-anova}, \cite{li-wa-li-11}  and \cite{sa-ja-ma-12} for the one-way ANOVA problem with unequal variances are same.


\bibliographystyle{apa}

\begin{thebibliography}{}

\bibitem[\protect\astroncite{Kotz and Nadarajah}{2004}]{kotz-na-04multi}
Kotz, S. and Nadarajah, S. (2004).
\newblock {\em Multivariate t-Distributions and Their Applications}.
\newblock Cambridge University Press, Cambridge.

\bibitem[\protect\astroncite{Li et~al.}{2011}]{li-wa-li-11}
Li, X., Wang, J., and Liang, H. (2011).
\newblock Comparison of several means: A fiducial based approach.
\newblock {\em Computational Statistics \& Data Analysis}, 55(5):1993--2002.

\bibitem[\protect\astroncite{Sadooghi-Alvandi et~al.}{2012}]{sa-ja-ma-12}
Sadooghi-Alvandi, S.~M., Jafari, A.~A., and Mardani-Fard, H.~A. (2012).
\newblock One-way {ANOVA} with unequal variances.
\newblock {\em Communications in Statistics-Theory and Methods},
  41(22):4200--4221.

\bibitem[\protect\astroncite{Sadooghi-alvandi et~al.}{2016}]{Sadooghi-ja-ma-16}
Sadooghi-alvandi, S.~M., Jafari, A.~A., and Mardani-Fard, H.~A. (2016).
\newblock Comparing several regression models with unequal variances.
\newblock {\em Communications in Statistics - Simulation and Computation},
  45(9):3190--3216.

\bibitem[\protect\astroncite{Tian et~al.}{2009}]{tian-ma-ve-09}
Tian, L., Ma, C., and Vexler, A. (2009).
\newblock A parametric bootstrap test for comparing heteroscedastic regression
  models.
\newblock {\em Communications in Statistics-Simulation and Computation},
  38(5):1026--1036.

\bibitem[\protect\astroncite{Xu and Li}{2018}]{xu-li-18}
Xu, J. and Li, X. (2018).
\newblock A fiducial p-value approach for comparing heteroscedastic regression
  models.
\newblock {\em Communications in Statistics - Simulation and Computation},
  47(2):420--431.

\bibitem[\protect\astroncite{Xu and Wang}{2008}]{xu-wa-08-anova}
Xu, L.-W. and Wang, S.-G. (2008).
\newblock A new generalized p-value for anova under heteroscedasticity.
\newblock {\em Statistics \& Probability Letters}, 78(8):963--969.

\end{thebibliography}

\end{document}